\documentclass[12pt]{amsart}
\usepackage{amsmath}
\usepackage{amssymb}
\usepackage[latin1]{inputenc}
\usepackage[T1]{fontenc}
\newtheorem{lemma}{Lemma}

\newtheorem{definition}{Definition}
\newtheorem{theorem}{Theorem}

\newcommand{\tr}{{\rm tr }}

\newcommand{\bra}{\langle}
\newcommand{\ket}{\rangle}
\newcommand{\vp}{\varphi}

\newcommand{\N}{\mathbb{N}}

\newcommand{\R}{\mathbb{R}}

\newcommand{\be}{\begin{equation}}
\newcommand{\eeq}{\end{equation}}
\newcommand{\bet}{\begin{equation*}}
\newcommand{\eeqt}{\end{equation*}}
\newcommand{\bea}{\begin{eqnarray}}
\newcommand{\eeqa}{\end{eqnarray}}
\newcommand{\beat}{\begin{eqnarray*}}
\newcommand{\eeqat}{\end{eqnarray*}}
\newcommand{\goesto}{\rightarrow}
\newcommand{\h}[1]{\mathcal{#1}}
\newcommand{\hil}{\mathcal{H}}

\newcommand{\cc}[1]{\overline{#1}}

\begin{document}
\title{A proof for the informational completeness of the rotated quadrature observables}
\author{J. Kiukas}
\address{Jukka Kiukas,
Department of Physics, University of Turku,
FIN-20014 Turku, Finland}
\email{jukka.kiukas@utu.fi}
\author{P. Lahti}
\address{Pekka Lahti,
Department of Physics, University of Turku,
FIN-20014 Turku, Finland}
\email{pekka.lahti@utu.fi}
\author{J.-P. Pellonpää}
\address{Juha-Pekka Pellonpää,
Department of Physics, University of Turku,
FIN-20014 Turku, Finland}
\email{juhpello@utu.fi}

\begin{abstract}
We give a new mathematically rigorous proof for the fact that,
when $S$ is a dense subset of $[0,2\pi)$,
the rotated quadrature operators $Q_\theta$,
$\theta\in S$,  of a single mode electromagnetic field constitute an informationally complete
set of observables.
\end{abstract}

\maketitle
\section{Introduction}
Since the pioneering works of Vogel and Risken \cite{VR1989} and Smithey {\em et al} \cite{Smithey1993}
the measurement of the rotated quadratures $Q_\theta$, $\theta\in[0,2\pi)$, has formed one of the major tools in the quantum state tomography,
see, e.g., the compilation \cite{kirja}. The basic idea behind the state reconstruction is well known: 
the inverse Radon tranform of the quadrature
measurement statistics allows one to reconstruct the Wigner function of the state in question and the Wigner function separates
states, so that the statistics uniquely determines the state. We do not question the validity of this argument. However, the existing
literature, which we are aware of, does not give a full justification that this procedure actually applies to all possible
states of a quantum system. Therefore, in this paper, we wish to give a direct proof of the fact that the set of quadrature
observables $Q_\theta$, $\theta\in S$, $S\subseteq[0,2\pi)$ is dense, is informationally complete, that is, the measurement outcome statistics 
$p^{Q_\theta}_T$
of these observables uniquely
determine the state $T$.

There is a beautiful group theoretical proof of the informational completeness of the observables $Q_\theta$, $\theta\in[0,2\pi)$
\cite{Gianni2000}. This proof builds on a general method of constructing informationally complete sets of observables 
using the theory of square-integrable representations of unimodular Lie groups. The result in question is then obtained as a
special application of this  theory to the Weyl-Heisenberg group. 
Due to the practical importance of the result, we give an alternative direct proof of it. The proof of this fact in  Section~\ref{proof}  forms the main body of this 
paper.

In the final Section~\ref{Phasespace} we comment on the measurability of the quadrature observables and we compare the tomography
based on their measurements on the one obtained from the phase space measurements.

\section{Basic notations and definitions}

Let $\hil$ be a complex separable Hilbert space, and $L(\hil)$ the set of bounded
operators on $\hil$, and $\h T(\hil)$ the set of trace class operators. We let $\|\cdot\|_1$
denote the trace norm of $\h T(\hil)$. (The operator norm of $L(\hil)$ will be denoted
by $\|\cdot\|$.) When $\hil$ is associated with a
quantum system (such as the single mode electromagnetic field), the states of the system are
being represented by positive operators $T\in \h T(\hil)$ with unit trace, density operators,
and the observables are associated with the normalized positive operator measures
defined on the Borel $\sigma$-algebra $\h B(\R)$ of the real line.\footnote{Normalized positive
operator measure is a map $E:\h B(\R)\to L(\hil)$ 
which is $\sigma$-additive in the weak operator topology, and has the property $E(\R)=I$, 
that is, for which $X\mapsto\tr[TE(X)]$ is a probability measure for each positive trace one operator $T$.}
Among them 
are the conventional von Neumann type of observables,
that is, projection valued measures $P:\h B(\R)\to L(\hil)$, or, equivalently, selfadjoint
operators in $\hil$.

The measurement outcome statistics of an observable $E:\h B(\R) \to L(\hil)$
in a state $T$ is given by the probability measure $X\mapsto \tr[TE(X)]=:p^E_T(X)$.

\begin{definition} A set $\h M$ of observables $E:\h B(\R)\to L(\hil)$ is
\emph{informationally complete}, if any two states $S$ and $T$ are equal
whenever $\tr[TE(X)] = \tr[SE(X)]$ for all $E\in \h M$ and $X\in \h B(\R)$.
\end{definition}

Thus, the informational completeness of a set $\h M$ of observables means that the totality of
the measurement outcome distributions $p^E_T$,  $E\in \h M$, determines
the state $T$ of the system. Clearly, a set $\h M$ of observables is informationally complete
if and only if $T=0$ whenever $T$ is a selfadjoint trace class operator with $\tr[TE(X)]=0$
for all $E\in \h M$ and $X\in \h B(\R)$. We will use this characterization in our proof.

Fix $\{|n\ket\mid n\in \N\}$ to be an orthonormal basis of $\hil$. (This is identified
with the photon number basis in the case where $\hil$ is associated with the single mode
electromagnetic field.) Here $\N = \{0,1,2,\ldots\}$. We will, without explicit indication, use the
coordinate representation, in which $\hil$ is represented as $L^2(\R)$ via the
unitary map $\hil\ni |n\ket\mapsto h_n\in L^2(\R)$, where $h_n$ is the $n$th Hermite function,
$$
h_n(x) = \frac{1}{\sqrt{2^nn!\sqrt{\pi}}}H_n(x)e^{-\frac 12 x^2},
$$
and $H_n$ is the $n$th Hermite polynomial.\footnote{Hermite polynomials are, of course, given
by the following recursion relation: $H_0(x)=1$, $H_1(x)=2x$ and $H_{n+1}(x)= 2xH_n(x)-2nH_{n-1}(x)$.}
Let $a$ and $a^*$ denote the usual raising and lowering operators associated with the above
basis of $\hil$; they are consider as being defined on their maximal domain
$$D(a)=D(a^*)= \Big\{\vp\in \hil \Big| \sum_{n=0}^\infty n|\bra \vp|n\ket|^2<\infty\Big\}.$$
Then define the operators $Q=\frac{1}{\sqrt{2}}\cc{(a^*+a)}$ and
$P=\frac{i}{\sqrt{2}}\cc{(a^*-a)}$, which, in the coordinate
representation are the usual multiplication and differentiation operators, respectively:
$(Q\psi)(x) = x\psi(x)$, $(P\psi)(x)= -i\frac{d\psi}{dx}(x)$.
(Here the bar stands for the closure of an
operator, so that e.g. $Q$ is the unique selfadjoint extension of the
essentially selfadjoint symmetric operator $\frac 1 {\sqrt 2}(a^*+a)$; see
\cite[Chapter IV]{Putnam} or \cite[Chapter 12]{Birman}
for details concerning the domains of these extensively studied
operators.) In the case of the electromagnetic field, $Q$ and $P$ are called the
quadrature amplitude operators of the field. In addition, $a=\frac{1}{\sqrt{2}}(Q+iP)$ and
$a^*=\frac{1}{\sqrt 2}(Q-iP)$ (see e.g.\ \cite[p.\ 73]{Putnam}). The Schwartz space $\h S(\R)$ of 
rapidly decreasing $C^\infty$-functions is included in $D(Q)\cap D(P)=D(a)$.

For a function $g:\R\to \R$, we let $g^{(k)}$ denote the $k$th derivative of $g$ (with $g^{(0)}=g$),
provided it exists. We need the following elementary commutation relations, which hold whenever
$\vp\in \h S(\R)$, and $g:\R\to \R$ is continuously differentiable and bounded:
\begin{eqnarray}
(g(Q)a^*-a^*g(Q))\vp &=& \tfrac{1}{\sqrt{2}}g^{(1)}(Q)\vp,\label{commutatorI}\\
(g(Q)a-ag(Q))\vp &=&-\tfrac{1}{\sqrt{2}}g^{(1)}(Q)\vp.\label{commutatorII}
\end{eqnarray}

Let $N$ denote the operator $a^*a$. It is selfadjoint on its natural domain
$$D(N)=\Big\{\vp\in \hil\Big|\sum_{n=0}^\infty n^2|\bra\vp|n\ket|^2<\infty\Big\}.$$
Now the phase shifting unitary operators are $e^{i\theta N}$, and we can define the
\emph{rotated quadrature observables} $Q_\theta$ by
$$
Q_\theta = e^{i\theta N}Qe^{-i\theta N}, \ \theta\in [0,2\pi).
$$
The spectral measure of $Q_\theta$ will be denoted by $P^{Q_\theta}:\h B(\R)\to L(\hil)$.

\section{The proof}\label{proof}

We need the so called \emph{Dawson's integral}
$$
{\rm daw}(x) = e^{-x^2}\int_0^xe^{t^2}dt
$$
(see e.g.\ \cite[pp.\ 298-299]{Abramowitz} or \cite[Chapter 42]{Atlas}).
The following lemma lists those properties of Dawson's integral that we are going to use.
Since the Dawson's integral has been studied extensively, they are probably well known.
However, as we were unable to find these results directly stated in the literature, we give
a proof in the Appendix; a reader familiar with the results may skip that proof.

\begin{lemma}\label{dawson}
\begin{itemize}
\item[(a)] ${\rm daw}:\R\to \R$ is a $C^\infty$-function, and
$$\lim_{x\goesto\pm\infty}{\rm daw}^{(k)}(x)=0 \ \ \text{for all } k\in \N.$$
\item[(b)]
$$
{\rm daw}^{(1)}(x) = \frac 12\sum_{n=0}^\infty \frac{(-1)^n n!}{2^n (2n)!}H_{2n}(x) \  \ \text{for all } x\in \R,
$$
where $H_n$ is the $n$th Hermite polynomial.
\end{itemize}
\end{lemma}
\begin{lemma}\label{lemma:alku}
There exists a $C^\infty$-function $f:\R\to \R$, such that
\begin{itemize}
\item[(i)] each derivative $f^{(n)}$ of $f$, $n=0,1,2\ldots$ is a bounded
function, and
\item[(ii)] $\bra n|f(Q)|n\ket = \delta_{n0}$, for all $n\in \N$, where $\delta$ is the Kronecker delta.
\end{itemize}
\end{lemma}
\begin{proof}
We put $f=2\,{\rm daw}^{(1)}$. According to Lemma \ref{dawson}, $f$ is
a $C^\infty$-function and (i) holds. Lemma \ref{dawson} (b) gives the expansion
\begin{equation}\label{hermiteseries}
f(x) = \sum_{n=0}^\infty \frac{(-1)^n n!}{2^n (2n)!}H_{2n}(x), \ x\in \R.
\end{equation}
Since the series in \eqref{hermiteseries} converges pointwise, we get
$$
f(x)h_0(x)
= \sum_{n=0}^\infty \frac{(-1)^n n!}{2^n (2n)!}\sqrt{2^{2n} (2n)!}\,h_{2n}
= \sum_{n=0}^\infty \frac{(-1)^n n!}{\sqrt{(2n)!}}h_{2n}(x)$$
for each $x\in \R$. In addition, $\sum_{n=0}^\infty (n!)^2/ (2n)!<\infty$,
so the series converges also in $L^2$-norm. Noting also that $H_n^2h_0\in L^2(\R)$, we can justify the following computation.
\begin{eqnarray*}
  \bra h_n |f(Q)h_n\ket &=& \frac{1}{2^nn!}\bra H_n^2h_0|fh_0\ket = 
  \frac{1}{2^nn!}\sum_{k=0}^\infty \frac{(-1)^k k!}{2^k (2k)!} \bra H_n^2h_0|\sqrt{2^{2k} (2k)!}h_{2k}\ket\\
   &=& \frac{1}{2^nn!}\sum_{k=0}^n \frac{(-1)^k k!}{2^k (2k)!}
\bra H_n^2h_0|H_{2k}h_0\ket\\
&=& \frac{1}{2^nn!}\sum_{k=0}^n \frac{(-1)^k k!}{2^k (2k)!\sqrt{\pi}}
\int_{\R} H_{2k}(x)(H_n(x))^2e^{-x^2}\, dx\\
&=& \frac{1}{2^nn!}\sum_{k=0}^n \frac{(-1)^k k!}{2^k (2k)!\sqrt{\pi}}
\frac{2^{k+n}\sqrt{\pi}(2k)!(n!)^2}{(n-k)!(k!)^2}=\sum_{k=0}^n\frac{(-1)^kn!}{(n-k)!k!}\\
&=& \sum_{k=0}^n (-1)^k\binom nk = \lim_{y\to1}(1-y)^n=\delta_{n0}.
\end{eqnarray*}
Here the finite sum after the third equality is obtained by noting that
$h_{2k}$ is orthogonal to $H_n^2h_0$ whenever $k>n$, which is due to the fact
that the latter function is a linear combination of Hermite functions
$h_0,\ldots ,h_{2n}$. The fifth equality follows from the formula 7.375(2) of
\cite[p.\ 838]{Gradshteyn}.
\end{proof}

Now we choose a function $f$ satisfying the conditions of Lemma
\ref{lemma:alku}. The function $f$ will remain fixed throughout the rest
of the paper.

\begin{lemma}\label{lemma:vaiheII}
$\bra n+k|f^{(k)}(Q)|n\ket = (-1)^k\sqrt{2^kk!}\,\delta_{0n}$ for all $k,n\in \N$.
\end{lemma}
\begin{proof}
We proceed by induction with respect to $k$; the initial step is provided
by condition (ii) of Lemma \ref{lemma:alku}. The induction assumption is that for some $k\in \N$,
the equality
$$\bra n+k|f^{(k)}(Q)|n\ket = (-1)^k\sqrt{2^kk!}\,\delta_{0n}$$
holds for all $n\in \N$. We must show that
$$\bra n+k+1|f^{(k+1)}(Q)|n\ket = (-1)^{k+1}\sqrt{2^{k+1}(k+1)!}\,\delta_{0n}, \ \ n\in \N.$$
But by using \eqref{commutatorI}, we get
\begin{align*}
 & \bra n+k+1|f^{(k+1)}(Q)|n\ket = \sqrt{2}\bra n+k+1|f^{(k)}(Q)a^*|n\ket
-\sqrt{2}\bra n+k+1|a^*f^{(k)}(Q)|n\ket\\
&= \sqrt{2(n+1)}\bra (n+1)+k|f^{(k)}(Q)|n+1\ket
-\sqrt{2(n+k+1)}\bra n+k|f^{(k)}(Q)|n\ket\\
&=-\sqrt{2(n+k+1)}(-1)^k\sqrt{2^kk!}\,\delta_{0n} = (-1)^{k+1}\sqrt{2^{k+1}(k+1)!}\,\delta_{0n}
\end{align*}
for all $n\in \N$ by the induction assumption.
(Note, in particular, that the first term in the expression following
the second equality is indeed zero by the induction assumption, because $n+1>0$.)
\end{proof}

\begin{lemma}\label{lemma:lopullinen}
$
\bra n+k|f^{(k+2l)}(Q)|n\ket = 2^l(-1)^k\sqrt{2^kl!(l+k)!}\,\delta_{ln}, \ k,l,n\in\N, \ n\geq l.$
\end{lemma}

\begin{proof} Now we use induction with respect to $l$, so the initial step $l=0$ is given by Lemma
\ref{lemma:vaiheII}. The induction assumption is that
$$\bra n+k|f^{(k+2l)}(Q)|n\ket = 2^l(-1)^k\sqrt{2^kl!(l+k)!}\,\delta_{ln}$$ holds for some $l\in \N$,
\emph{all} $k\in \N$, and \emph{all} $n\geq l$. We have to show that this holds also when $l$ is replaced
by $l+1$. Accordingly, let $k\in \N$, and $n\in \N$ with $n\geq l+1$.
Using the commutation relation \eqref{commutatorII} and the induction assumption, we get
\begin{align*}
& \bra n+k|f^{(k+2(l+1))}(Q)|n\ket = -\sqrt{2}\bra n+k|f^{(k+2l+1)}(Q)a|n\ket +\sqrt{2}\bra n+k|af^{(k+2l+1)}(Q)|n\ket\\
&= -\sqrt{2n}\bra (n-1)+(k+1)|f^{((k+1)+2l)}(Q)|n-1\ket
+\sqrt{2(n+k+1)}\bra n+(k+1)|f^{((k+1)+2l)}(Q)|n\ket\\
&= -\sqrt{2n}2^l(-1)^{k+1}\sqrt{2^{k+1}l!(l+k+1)!}\,\delta_{l,n-1}
= (-1)^k\sqrt{2(l+1)}2^l\sqrt{2}\sqrt{2^{k}l!(l+1+k)!}\,\delta_{l+1,n}\\
& = 2^{l+1}(-1)^k\sqrt{2^k(l+1)!(l+1+k)!}\,\delta_{l+1,n}.
\end{align*}
Here the induction assumption was applied to the first term following the second equality
with $n$ and $k$ replaced by $n-1$ and $k+1$, and to the second term with $n$ and $k$ replaced by
$n$ and $k+1$. Here it is essential to note that $n\geq l+1>l$, which makes the second term zero.
\end{proof}
 
In order to construct the proof for the informational completeness of the quadratures, we
need some additional tools. First define, for each fixed $k\in \N$ and $X\in \h B(\R)$,
$$
V^k(X):=\int_0^{2\pi} e^{-ik\theta} P^{Q_\theta}(X)\, d\theta\in L(\hil),
$$
where the integral is to be understood in the $\sigma$-weak operator topology.
Indeed, for each trace class operator $T$, the map
$\theta\mapsto \tr[TP^{Q_\theta}(X)]=\tr[T e^{i\theta N}P^{Q}(X)e^{-i\theta N}]$ is continuous,\footnote{
This can easily be seen by e.g.\ considering a positive trace class operator $T$, using its spectral resolution and applying the strong continuity of the map $\theta\mapsto e^{i\theta N}$.}
and $|\tr[TE^{Q_\theta}(X)]|\leq \|T\|_1\|P^{Q_\theta}(X)\|\leq \|T\|_1$, so
the integral exists in the $\sigma$-weak sense, and represents a bounded operator, with
$\|V^k(X)\|\leq 2\pi$.

Next, notice that for each $k\in \N$, the map $X\mapsto V^k(X)$ is an operator measure, that is,
$\sigma$-additive in the weak operator topology. In fact, if $(X_n)_{n\in \N}$ is a sequence of
mutually disjoint sets in $\h B(\R)$, then for any $l\in \N$,
$$\left|\sum_{n=0}^l e^{-ik\theta}\bra \vp|P^{Q_\theta}(X_n)\vp\ket\right|\leq \sum_{n=0}^l \bra \vp|P^{Q_\theta}(X_n)\vp\ket\leq \bra \vp|P^{Q_\theta}(\cup_{n=0}^\infty X_n)\vp\ket \leq \|\vp\|^2,
$$
so the dominated convergence theorem can be applied to get
\begin{align*}
\sum_{n=0}^\infty \bra\vp|V^k(X_n)\vp\ket &= \sum_{n=0}^\infty \int_0^{2\pi} e^{-ik\theta} \bra \vp|P^{Q_\theta}(X_n)\vp\ket \, d\theta = \int_0^{2\pi} e^{-ik\theta} \bra \vp|P^{Q_\theta}(\cup_{n=0}^\infty X_n)\vp\ket \, d\theta\\
& = \bra \vp|V^k(\cup_{n=0}^\infty X_n)\vp\ket.
\end{align*}

Let $g:\R\to \R$ be any bounded Borel function. Then the operator integral $V^k[g]=\int g \,dV^k$
can be defined in the $\sigma$-weak sense, as a bounded operator.
This follows from the fact that $g$ is bounded and
$$|V^k_T|(\R)\leq 4\sup_{X\in \h B(\R)}|\tr[TV^k(X)]|\leq 8\pi \|T\|_1$$
for any trace class operator $T$, with $|V^k_T|$ denoting the total variation of the complex measure
$V^k_T=\tr[TV^k(\cdot)]$. (The map $X\mapsto \tr[TV^k(X)]$ is a complex measure, because the weak
and $\sigma$-weak operator topologies coincide in a norm-bounded set.)

\begin{lemma}\label{explicit}
For any bounded function $g:\R\to \R$, and a trace class operator $T$, we have
$$
\tr[TV^k[g]] = 2\pi\sum_{n=0}^\infty \bra n|T|n+k\ket \,\bra n+k|g(Q)|n\ket.
$$ 
\end{lemma}
\begin{proof}
First note that
$\bra m|V^k[g]|n\ket = \tr[|n\ket\bra m|V^k[g]] = \int g \, dV^k_{|n\ket\bra m|}$ by definition.
On the other hand, for each $X\in \h B(\R)$, we get
$$V^k_{|n\ket\bra m|}(X)=\int_0^{2\pi} e^{-ik\theta} \bra m|P^{Q_\theta}(X)|n\ket\, d\theta = 
\bra m|P^{Q}(X)|n\ket \int_0^{2\pi} e^{-ik\theta} e^{i\theta(m-n)}\, d\theta.
$$
This equals $2\pi\bra n+k|P^{Q}(X)|n\ket$ if $m=n+k$ and zero otherwise. Hence,
$$\bra n+k|V^k[g]|n\ket = 2\pi \int g\, dP^Q_{|n\ket\bra n+k|}= 2\pi \bra n+k|g(Q)|n\ket,$$
and $\bra m|V^k[g]|n\ket=0$ whenever $m\neq n+k$. Thus,
$$
\tr[TV^k[g]]= \sum_{n=0}^\infty \bra n|TV^k[g]|n\ket
= \sum_{n=0}^\infty\sum_{m=0}^\infty \bra n|T|m\ket\bra m|V^k[g]|n\ket
= 2\pi\sum_{n=0}^\infty \bra n|T|n+k\ket \bra n+k|g(Q)|n\ket.
$$ 
\end{proof}

Now we are ready to prove the actual result of the paper.
\begin{theorem} Let $S$ be a dense subset of $[0,2\pi)$. The set $\{P^{Q_\theta}\mid \theta\in S\}$ of observables is
informationally complete.
\end{theorem} 
\begin{proof}
Let $T\in L(\hil)$ be a selfadjoint trace class operator, such that
$\tr[TP^{Q_\theta}(X)]=0$ for all $X\in \h B(\R)$ and $\theta\in S$.
Since $\theta\mapsto\tr[TP^{Q_\theta}(X)]$ is continuous and $S$ is dense it follows that
$\tr[TP^{Q_\theta}(X)]=0$ for all $X\in \h B(\R)$ and $\theta\in[0,2\pi)$.

Let $k\in \N$ be fixed. By the definition of the operator measure $V^k$, the assumption implies
that $V^k_T(X)=\tr[TV^k(X)]=
\int_0^{2\pi}e^{-ik\theta}\tr[TP^{Q_\theta}(X)]d\theta
=0$ for all $X\in \h B(\R)$. From the definition of the $\sigma$-weak integral $V^k[g]=\int g\, dV^k$,
it follows that $\tr[TV^k[g]]=\int g\, dV^k_T=0$ for any bounded Borel function $g:\R\to \R$.
Hence, by using Lemma
\ref{explicit}, we get
\begin{equation}\label{summation}
\sum_{n=0}^\infty \bra n|T|n+k\ket\,\bra n+k|g(Q)|n\ket=0 
\end{equation}
for any bounded Borel function $g:\R\to\R$.
We show by induction with respect to $n$ that $\bra n|T|n+k\ket = 0$ for all $n\in \N$.
First, put $g=f^{(k)}$ in \eqref{summation} (recall that $f$ was a function satisfying the assumptions of
Lemma \ref{lemma:alku}). By Lemma \ref{lemma:lopullinen}, this gives $\bra 0|T|k\ket = 0$,
which proves the initial step. The induction assumption is that for some $m\in \N$,
$\bra n|T|n+k\ket = 0$ for all $n\in \N$, $n\leq m$. We have to show that this implies
$\bra m+1|T|(m+1)+k\ket = 0$. By the induction assumption and \eqref{summation}, we
have
$$
\sum_{n=m+1}^\infty \bra n|T|n+k\ket\bra n+k|f^{(k+2(m+1))}(Q)|n\ket=0,
$$
where we have simply put $g=f^{(k+2(m+1))}$, which is again a bounded Borel function.
But, according to Lemma \ref{lemma:lopullinen},
$$\bra n+k|f^{(k+2(m+1))}(Q)|n\ket=0, \ \ n> m+1,$$ and
$$\bra (m+1)+k|f^{(k+2(m+1))}(Q)|m+1\ket\neq 0,$$
so that necessarily $\bra m+1|T|(m+1)+k\ket=0$. This completes the induction proof.

We have thus established that $\bra l|T|l+k\ket = 0$ for all $l,k\in \N$. Since $T$ is selfadjoint, this
implies that $T=0$, and the proof is complete.
\end{proof}

Note that the set $S$ in the previous theorem can be chosen to be countable
(e.g. $S= [0,2\pi)\cap \mathbb Q$).
Hence, in principle, it suffices to measure a \emph{sequence} of quadrature
observables in order to determine the state of the system.

Though obvious, it may be worth to note that the quadrature observables $Q_\theta$ are not informationally complete in the sense of statistical expectation,
that is, the numbers $\tr[TQ_\theta]$, $\theta\in[0,2\pi)$, do not, in general, determine the state $T$; for instance, the number states $|n\ket$ are
indistinguishable by the expectations, $\bra n|Q_\theta|n\ket=0$ for all $n$ and for all $\theta$. 
 
\section{Wigner function {\em vs.} phase space distributions}\label{Phasespace}
According to the result in the preceding section, the quadrature observables
$Q_\theta$, $\theta\in S$ ($S\subseteq[0,2\pi)$ is dense) constitute  an informationally complete set of observables, i.e.
the measurement statistics $p^{Q_\theta}_T$, $\theta\in S$, determine uniquely
the state $T$ of the quantum system. The question of experimental implementation of these observables is thus of utmost importance.

The balanced homodyne detection with a strong auxiliary field is a well developed method of experimental quantum physics, and this
method is known to yield the measurement statistics of the quadrature observable $Q_\theta$, depending on the phase $e^{i\theta}$
of the (one-mode) auxiliary field. The heuristic physical argument behind this method is equally well known, see, e.g.\
\cite{Leonhard,Haroche}, the detailed mathematical justification being, however, more involved. 

If $|z\ket$, $z=re^{i\theta}$, is the coherent state of the (one-mode) auxiliary field, then the actually measured observable 
in the balanced homodyne detection is given by a semispectral measure $E^z$, whose first moment operator $E^z[1]$
is an extension of the restriction of the quadrature operator $Q_\theta$ on the domain $D(a)$ of the signal mode 
 operator $a$, and whose noise operator $E^z[2]-E^z[1]^2$ equals with the operator $\frac 12 r^{-2}N$
 (where $N=a^*a$).
Clearly, these results suggest that the high amplitude limit of $E^z$ is the spectral measure $P^{Q_\theta}$ of $Q_\theta$,
notably since the spectral measures are known to be exactly those semispectral measures whose noise operators 
equal zero.  There is, indeed, a rigorous quantum mechanical proof of the fact that in the high amplitude limit the observable
$E^z$ tends to the spectral measure of $Q_\theta$, though the actual meaning of this limit requires more caution \cite{Jukka1}.

The balanced homodyne detection scheme thus allows one to collect the statistics of the quadrature observables.
Assuming that the inverse Radon transform can be applied to the collection of distributions $p^{Q_\theta}_T$,
$\theta\in[0,2\pi)$, one obtains the Wigner function $f^W_T$ of the state $T$,
which is a unique representation of $T$. Clearly, there is no need to use  the inverse Radon transform nor the Wigner function, since the set of
quadrature observables is, in itself, informationally complete. 
In any case,  this method of state reconstruction is unnecessarily complicated, 
since it requires that one measures infinite number of different observables $Q_\theta$, $\theta\in S$.

It is well-known that a beam splitter combined with the $Q_0$ and $Q_{{\pi}/2}$ -sensitive detectors  constitutes a measurement of
a phase space observable $E^D$ of the signal mode, with the generating density operator $D$ depending on the state of the (one-mode)
auxiliary field, see, e.g. \cite[VII.3.7]{OQP}. Using the balanced homodyne detection realization for the 
$Q_0$ and $Q_{{\pi}/2}$ -sensitive detectors one indeed obtains, in a rigorous quantum mechanical sense, an eight-port homodyne detector
realization of an arbitrary phase space observable $E^D$ \cite{Jukka2}. A phase space observable $E^D$ is known to be informationally
complete whenever the generating density operator $D$ 
is such that $\tr[W_{qp}D]\ne 0$ for almost all
phase space points $(q,p)\in\R^2$; here $W_{qp}$ is the Weyl operator, see, e.g.\ \cite{Prugo}. 
In particular, if the auxiliary field is idle, then the measurement statistic obtained from the eight-port homodyne detector is simply
the Husimi function $f^H_T$ of the signal state $T$. Like the Wigner function,  the Husimi function separates states, that is,
the phase space observable $E^{D}$, $D=|0\ket\bra 0|$, is informationally complete. It is obvious that
the state reconstruction via the eight-port homodyne detector is a huge simplification 
when compared with the reconstruction using only  the balanced homodyne detection technique.

\section*{Appendix: The proof of Lemma \ref{dawson}}

To begin the proof of the lemma, we first note that the Dawson's integral is clearly a
$C^\infty$-function. We prove (a) by using the expansion
$$
{\rm daw}(x) = \frac{1}{2x}\sum_{j=0}^{n-1}\frac{(2j-1)!!}{(2x^2)^j}
+\frac{(2n-1)!!}{2^n}R_n(x),
$$
where
$$
R_n(x)=\frac{e^{-x^2}}{x^{2n-1}}\sum_{j=0}^\infty \frac{x^{2j}}{j!(2j-2n+1)}
$$
and $n=1,2,3,\ldots$ (see \cite[equation 42:6:5, p.\ 407]{Atlas}).
Since each derivative is either even or odd, it suffices to consider the
limit $x\goesto \infty$. Clearly, any derivative of the first part tends to
zero at this limit, for any choice of $n$. As for $R_n$, it is easy to see
that for any given $k$, we get $\lim_{x\goesto\infty}R_n^{(k)}(x)=0$
for suffiently large $n$. Indeed, let $$S_n(y)=\sum_{j=0}^\infty \frac{y^j}{j!(2j-2n+1)},$$ so that
$R_n(x) = \frac{e^{-x^2}}{x^{2n-1}}S_n(x^2)$. Since the power series
$S_n(y)$ is clearly convergent for all $y\in \R$, it can be differentiated
$k$ times (for any $k\in \N$) to get
$$S_n^{(k)}(y)=\sum_{j=k}^\infty \frac{y^{j-k}}{(j-k)!(2j-2n+1)}.$$
From this we see that $|S_n^{(k)}(y)|\leq e^{y}$ for any $y>0$ and $n\in \N$.
Now for a fixed $k\in \N$, put e.g. $n=k+1$. Since
$R^{(k)}_n(x)$ is clearly a finite sum of terms of the form
$A_{l,l'}\frac{e^{-x^2}}{x^{2n-1-l}}S_n^{(l')}(x^2)$,
with $-k\leq l\leq k$, $0\leq l'\leq k$, and $A_{l,l'}$ a
constant, it follows that $\lim_{x\goesto\infty}R_n^{(k)}(x)=0$. The proof of (a) is complete.

To prove (b), consider the series
\begin{equation}\label{dawsonhermite}
\frac 12\sum_{n=0}^\infty \frac{(-1)^n n!}{2^n (2n)!}H_{2n}(x).
\end{equation}
We first note that the well-known relation $\frac{d}{dx}H_{l}(x)=2lH_{l-1}(x)$, $l=1,2,3\ldots$, implies
\begin{equation}\label{derivative}
\frac{d^m}{dx^m} H_{2n}(x) = 2n(2n-1)\cdots (2n-m+1)2^m H_{2n-m}(x), \ \ m\leq 2n.
\end{equation}
Using the estimate $|H_{n}(x)|\leq e^{\frac 12 x^2}K 2^{\frac 12 n}\sqrt{n!}$, where $K>0$ is a constant
(\cite[22.14.17, p.\ 787]{Abramowitz}), we get
$$
\left|\frac{(-1)^n n!}{2^n (2n)!}\frac{d^m}{dx^m}H_{2n}(x)\right|\leq  2^{\frac m2}K\frac{n!\,2n(2n-1)\cdots (2n-m+1)\sqrt{(2n-m)!}}{(2n)!}e^{\frac 12x^2} \leq 2^{\frac m2}K \frac{(2n)^mn!}{\sqrt{(2n)!}}e^{\frac 12 x^2}
$$
for all $n,m\in \N$, $m\leq 2n$.
Now $\sum_{n=0}^\infty\frac{(2n)^mn!}{\sqrt{(2n)!}}<\infty$
for any $m\in \N$ by the ratio test, which shows that the series obtained by
differentiating \eqref{dawsonhermite} $m$ times term by term converges uniformly in bounded intervals.
Consequently, that series represents the $m$th derivative of the original series
\eqref{dawsonhermite}, the latter converging to some $C^\infty$-function $F:\R\to\R$ uniformly in
bounded intervals.

By using again the formula $\frac{d}{dx} H_{l}(x) = 2lH_{l-1}(x)$, we get
\begin{eqnarray}
2F^{(2m)}(x) &=& 2^{2m}\sum_{n=m}^\infty \frac{n!(-1)^n H_{2(n-m)}(x)}{2^n(2(n-m))!};\label{evenpw}\\
2F^{(2m+1)}(x) &=& 2^{2m+1}\sum_{n=m+1}^\infty \frac{n!(-1)^n H_{2(n-m)-1}(x)}{2^n(2(n-m)-1)!}.\label{oddpw}
\end{eqnarray}

In order to calculate the MacLaurin series of $F$, we need the expansion
\begin{equation}\label{generalpower}
(1-x)^{-(m+1)} = \sum_{n=0}^\infty\binom{m+n}{n} x^n, \ \ -1<x< 1,
\end{equation}
which can be obtained from the binomial series (see 3.6.9. in \cite[p.\ 15]{Abramowitz}).

Now, using \eqref{evenpw}, the identity $H_{2(n-m)}(0)=(-1)^{n-m}(2(n-m))!/(n-m)!$, $n\geq m$
(22.4.8 in \cite[p.\ 777]{Abramowitz}), and \eqref{generalpower} with $x=\frac 12$, we get
\begin{eqnarray*}
2F^{(2m)}(0) &=& 2^{2m}(-1)^m\sum_{n=m}^\infty\frac{n!}{2^n(n-m)!} = 2^m(-1)^m\sum_{n=0}^\infty\frac{(n+m)!}{2^nn!}\\
&=& 2^m(-1)^mm!\sum_{n=0}^\infty\binom{n+m}{n}\frac{1}{2^n} = 2^{2m+1}(-1)^mm!.
\end{eqnarray*}
Since $F^{(2m+1)}(0)= 0$ by \eqref{oddpw} and the fact that $H_{l}(0)=0$ for odd $l$, we get the following
MacLaurin series for $F$:
\begin{equation}\label{taylorseriesforh}
\sum_{m=0}^\infty\frac{(-1)^mm!}{(2m)!}(2x)^{2m}, \ \ x\in \R.
\end{equation}
This is exactly the MacLaurin series for the first derivative of the Dawson's integral, since
$$
{\rm daw}(x) = \sum_{m=0}^\infty \frac{(-1)^mm!2^{2m}}{(2m+1)!}x^{2m+1}
$$
(see e.g. \cite[p.\ 406]{Atlas}).

In order to prove that the series \eqref{taylorseriesforh} actually converges to $F$ (pointwise for
all $x\in \R$), we have to show
that for each $x\in \R$, the corresponding remainder $F^{(k)}(\xi_k)x^k/k!$,
where $\xi_k\in [-|x|,|x|]$, goes to zero as $k\goesto\infty$.
By applying the estimate $|H_l(y)|\leq Ke^{\frac 12 y^2} 2^{l/2}\sqrt{l!}$ to \eqref{evenpw} and \eqref{oddpw}, we get
\begin{eqnarray*}
2|F^{(2m)}(y)| &\leq& e^{\frac 12 y^2}2^mK\sum_{n=0}^\infty \frac{(n+m)!}{\sqrt{(2n)!}}
= Ke^{\frac 12 y^2}2^mm!\sum_{n=0}^\infty \binom{n+m}{n}\frac{n!}{\sqrt{(2n)!}};\\
2|F^{(2m+1)}(y)| &\leq& e^{\frac 12 y^2}\sqrt{2}2^mK\sum_{n=0}^\infty \frac{(n+m+1)!}{\sqrt{(2n+1)!}}
= \sqrt{2}Ke^{\frac 12 y^2}2^m(m+1)!\sum_{n=0}^\infty \binom{n+m+1}{n}\frac{n!}{\sqrt{(2n+1)!}}
\end{eqnarray*}
for all $y\in \R$. It is easy to see by induction that $n!/\sqrt{(2n)!}\leq n/2^{n-1}$ for all $n\geq 1$; using this, as
well as the relation $(m+1)2^{m+1}=\sum_{n=0}^\infty\binom{m+n}{n}\frac{n}{2^n}$ which is obtained by
differentiating \eqref{generalpower} and putting $x=\frac 12$, we conclude that
$2|F^{(2m)}(y)|\leq 4Ke^{\frac 12 y^2}2^{2m}(m+1)!$, and $2|F^{(2m+1)}(y)|\leq 4\sqrt{2}Ke^{\frac 12 y^2}2^{2m+1}(m+2)!$ for all $y\in \R$.
Consequently,
\begin{eqnarray*}
\frac{|F^{(2m)}(\xi)x^{2m}|}{(2m)!} \leq 2K e^{\frac 12 x^2}\frac{(2|x|)^{2m}(m+1)!}{(2m)!},\\
\frac{|F^{(2m+1)}(\xi)x^{2m+1}|}{(2m+1)!} \leq 2\sqrt{2}K e^{\frac 12 x^2}\frac{(2|x|)^{2m+1}(m+2)!}{(2m+1)!},
\end{eqnarray*}
whenever $x\in \R$, $\xi\in [-|x|,|x|]$ and $m\in \N$. It is easy to see that for each fixed $x\in \R$, the right hand sides of
the inequalities tend to zero as $m\goesto \infty$. Hence, we have shown that
$$
F(x)= \sum_{n=0}^\infty\frac{(-1)^mm!}{(2m)!}(2x)^{2m}= {\rm daw}^{(1)}(x), \ \ x\in \R.
$$
This completes the proof of Lemma \ref{dawson}.

\

\noindent {\bf Acknowledgment.} One of the authors (J. K.) was supported by Emil Aaltonen Foundation.


\begin{thebibliography}{99}




\bibitem{Abramowitz} M. Abramowitz, I. A. Stegun (eds.), {\em Handbook of Mathematical Functions}, National Bureau of Standards, Applied Mathematics Series - 55, Tenth printing with corrections, 1972.
\bibitem{Prugo} S.T. Ali, E. Prugove\v cki, Classical and quantum statistical mechanics in a common Liouville space,
{\em Physica} {\bf 89A} (1977) 501-521.
\bibitem{Birman} M. S. Birman, M. Z. Solomjak, {\em Spectral Theory of Self-Adjoint Operators in Hilbert Space},
D. Reidel, Dordrecht, 1987.
\bibitem{OQP} P. Busch, M. Grabowski, P. Lahti, {\em Operational Quantum Physics}, {\bf LNP m31}, Springer, 1995, 2nd corrected
printting, 1997.
\bibitem{Gianni2000} G. Cassinelli, G.M. D'Ariano, E. De Vito, A. Levrero, {\em Group theoretical quantum tomography},
{\em J. Math. Phys.} {\bf 41} (2000) 7940-7951.
\bibitem{Gradshteyn} I. S. Gradshteyn, I. M. Ryzhnik, {\em Table of Integrals, Series, and Products},
Corrected and Enlarged Edition, Academic Press, Inc., Orlando, 1980.
\bibitem{Haroche} S. Haroche, J.-M. Raimond, {\em Exploring the Quantum}, Oxford UP, 2006.
\bibitem{Jukka1} J. Kiukas, P. Lahti, On the moment limit of quantum observables, with an application to the balanced
homodyne detection, {\em J. Mod. Opt.}, in press (2007).
\bibitem{Jukka2} J. Kiukas, P. Lahti, A note on the measurement of phase space observables with an eight-port homodyne detection,
{\tt arXiv:0708.4094v1 [quant-ph]}. 
\bibitem{Leonhard} U. Leonhard, {\em Measuring the Quantum State of Light}, Cambridge UP, 1997.
\bibitem{kirja} M.G.A. Paris, J.\v Reh\' a\v cek (eds), {\em Quntum State Estimation}, Springer, 2004.
\bibitem{Putnam} C. R. Putnam, {\em Commutation Properties of Hilbert Space Operators and Related Topics}, Springer-Verlag, Berlin, 1967.
\bibitem{Smithey1993} D.T. Smithey, M. Beck, M.G.Raymer, A. Faridina, Measurement of the Wigner distributuon and the density 
matrix of a light mode using optical homodyne tomography: application to squeezed states and the vacuum,
{\em Phys. Rev. Lett.} {\bf 70} (1993) 1244-1247.
\bibitem{Atlas} J. Spanier, K. B. Oldham, {\em An Atlas of Functions}, Hemisphere Publishing Company, 1987.
\bibitem{VR1989} K. Vogel, H. Risken, Determination of quasiprobabilitydistributions in terms of probability distributions for
the rotated quadrature phase, {\em Phys. Rev. A} {\bf 40} (1989) 2847-2849.
\end{thebibliography}
\end{document}